\documentclass[11pt,a4paper]{article}

\usepackage{geometry}
\usepackage[utf8x]{inputenc}
\usepackage{amsmath, amsthm}
\usepackage{graphicx,amssymb,textcomp,array,amsmath}
\usepackage{algpseudocode} 
\algtext*{EndWhile}
\algtext*{EndIf}
\algtext*{EndFor}
\usepackage{algorithm}

\newcommand{\MC}{M_\times}

\newcommand{\MOPT}{M^*}
\newcommand{\bt}{\lambda}
\newcommand{\btopt}{\lambda^*}
\newcommand{\cw}{cw}

\newcommand{\require}{\textbf{Input: }}
\newcommand{\ensure}{\textbf{Output: }}
\newtheorem{lemma}{Lemma}
\newtheorem{corollary}{Corollary}
\newtheorem{theorem}{Theorem}

\title{Approximating the Bottleneck Plane Perfect Matching\\of a Point Set}

\author{
A. Karim Abu-Affash \thanks{Software Engineering Department, Shamoon College of Engineering, Beer-Sheva 84100, Israel, {\tt abuaa1@sce.ac.il.}} 
\and 
Ahmad Biniaz
\thanks{School of Computer Science, Carleton University, 
                    Ottawa, Canada. Research supported by NSERC.}
\and 
Paz Carmi\thanks{Department of Computer Science, Ben-Gurion University of the Negev, Beer-Sheva 84105, Israel, {\tt carmip@cs.bgu.ac.il.}}
\and
Anil Maheshwari\footnotemark[2]
\and 
Michiel Smid\footnotemark[2]
}
\date{\today}

\begin{document}

\maketitle

\begin{abstract}
A bottleneck plane perfect matching of a set of $n$ points in $\mathbb{R}^2$ is defined to be a perfect non-crossing matching that minimizes the length of the longest edge; the length of this longest edge is known as {\em bottleneck}. The problem of computing a bottleneck plane perfect matching has been proved to be NP-hard. 
We present an algorithm that computes a bottleneck plane matching of size at least $\frac{n}{5}$ in $O(n \log^2 n)$-time.
Then we extend our idea toward an $O(n\log n)$-time approximation algorithm which computes a plane matching of size at least $\frac{2n}{5}$ whose edges have length at most $\sqrt{2}+\sqrt{3}$ times the bottleneck.

{\bf Key words:} plane matching, bottleneck matching, geometric graph, unit disk graph, approximation algorithm.

\end{abstract}

\section{Introduction}
We study the problem of computing a bottleneck non-crossing matching of points in the plane.
For a given set $P$ of $n$ points in the plane, where $n$ is even, let $K(P)$ denote the complete Euclidean graph with vertex set $P$. The {\em bottleneck plane matching} problem is to find a perfect non-crossing matching of $K(P)$ that minimizes the length of the longest edge. We denote such a matching by $\MOPT$. The bottleneck, $\btopt$, is the length of the longest edge in $\MOPT$. The problem of computing $\MOPT$ 
has been proved to be NP-hard \cite{Abu-Affash2014}. 
Figure 1 in \cite{Abu-Affash2014} and \cite{Carlsson2010} shows that the longest edge in the minimum weight matching (which is planar) can be unbounded with respect to $\btopt$. On the other hand the weight of the bottleneck matching can be unbounded with respect to the weight of the minimum weight matching, see Figure \ref{weight}.

\begin{figure}[ht]
  \centering
\setlength{\tabcolsep}{0in}
  $\begin{tabular}{cc}
 \multicolumn{1}{m{.5\columnwidth}}{\centering\includegraphics[width=.3\columnwidth]{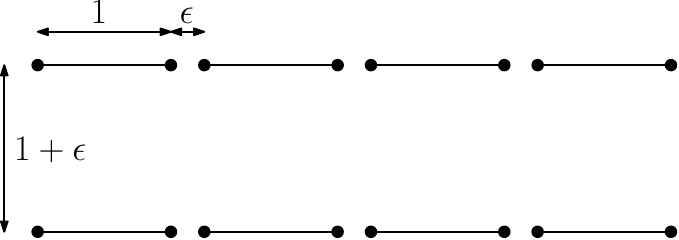}}
&\multicolumn{1}{m{.5\columnwidth}}{\centering\includegraphics[width=.3\columnwidth]{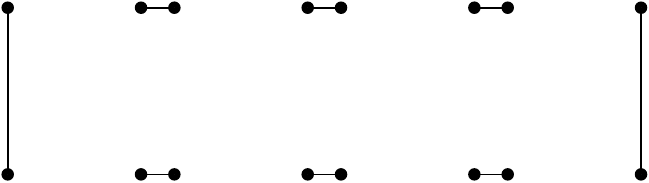}}\\
(a) & (b)
  \end{tabular}$
  \caption{(a) bottleneck matching, (b) minimum weight matching.}
\label{weight}
\end{figure}

Matching and bottleneck matching problems play an important role in graph theory, and thus, they have been studied extensively,  e.g.,~\cite{Abu-Affash2014, Aloupis2013, Chang1992, Efrat2001, Efrat2000, Vaidya1989, Varadarajan1998}. Self-crossing configurations are often undesirable and may even imply an error condition; for example, a potential collision between moving objects, or inconsistency in a layout of a circuit.
In particular, non-crossing matchings are especially important in the context of VLSI circuit layouts~\cite{Lengauer1990} and operations research.

\subsection{Previous Work}

It is desirable to compute a perfect matching of a point set in the plane, which is optimal with respect to some criterion such as: (a) {\em minimum-cost matching} which minimizes the sum of the lengths of all edges; also known as {\em minimum-weight matching} or {\em min-sum matching}, and (b) {\em bottleneck matching} which minimizes the length of the longest edge; also known as {\em min-max matching} \cite{Efrat2000}. For the minimum-cost matching, Vaidya \cite{Vaidya1989} presented an $O(n^{2.5}\log^4 n)$ time algorithm, which was improved to $O(n^{1.5}\log^5 n)$ by Varadarajan \cite{Varadarajan1998}. As for bottleneck matching, Chang et al. \cite{Chang1992} proved that such kind of matching is a subset of 17-RNG (relative neighborhood graph). They presented an algorithm, running in $O(n^{1.5}\sqrt{\log n})$-time to compute a bottleneck matching of maximum cardinality. The matching computed by their algorithm may be crossing.
Efrat and Katz \cite{Efrat2000} extended the result of Chang et al. \cite{Chang1992} to higher dimensions. They proved that a bottleneck matching in any constant dimension can be computed in $O(n^{1.5}\sqrt{\log n})$-time under the $L_\infty$-norm. 

Note that a plane perfect matching of a point set can be computed in $O(n\log n)$-time, e.g., by matching the two leftmost points recursively.

Abu-Affash et al.~\cite{Abu-Affash2014} showed that the bottleneck plane perfect matching problem is NP-hard and presented an algorithm that computes a plane perfect matching whose edges have length at most $2\sqrt{10}$ times the bottleneck, i.e., $2\sqrt{10}\btopt$. They also showed that this problem does not
admit a PTAS (Polynomial Time Approximation Scheme), unless P=NP. Carlsson et al.~\cite{Carlsson2010} showed that the bottleneck plane perfect matching problem for a Euclidean bipartite complete graph is also NP-hard.

\subsection{Our Results}

The main results of this paper are summarized in Table \ref{table1}. We use the unit disk graph as a tool for our approximations. First, we present an $O(n\log n)$-time algorithm in Section \ref{UDG}, that computes a plane matching of size at least $\frac{n-1}{5}$ in a connected unit disk graph. Then in Section \ref{bottleneck-five-over-two} we describe how one can use this algorithm to obtain a bottleneck plane matching of size at least $\frac{n}{5}$ with edges of length at most $\btopt$ in $O(n\log^2 n)$-time.
In Section \ref{bottleneck-five-over-four} we present an $O(n\log n)$-time approximation algorithm that computes a plane matching of size at least $\frac{2n}{5}$ whose edges have length at most $(\sqrt{2}+\sqrt{3})\btopt$. Finally we conclude this paper in Section \ref{conclusion}.

\begin{table}
\centering
\caption{Summary of results.}
\label{table1}
    \begin{tabular}{|l|c|c|c|}
         \hline
             Algorithm &time complexity&bottleneck ($\bt$)    & size of matching  \\ \hline
             Abu-Affash et al. \cite{Abu-Affash2014}&$O(n^{1.5}\sqrt{\log n})$  & $2\sqrt{10}\btopt$ & $n/2$ \\
             Section \ref{bottleneck-five-over-two} &$O(n \log^2 n)$& $\btopt$ & $n/5$ \\
             Section \ref{bottleneck-five-over-four}&$O(n\log n)$ & $(\sqrt{2}+\sqrt{3})\btopt$ & $2n/5$\\
         \hline
    \end{tabular}
\end{table}

\section{Preliminaries}
\label{preliminaries}

Let $P$ denote a set of $n$ points in the plane, where $n$ is even, and let $K(P)$ denote the complete Euclidean graph over $P$. A {\em matching}, $M$, is a subset of edges of $K(P)$ without common vertices. Let $|M|$ denote the {\em cardinality} of $M$, which is the number of edges in $M$. $M$ is a {\em perfect matching} if it covers all the vertices of $P$, i.e., $|M|=\frac{n}{2}$. The {\em bottleneck} of $M$ is defined as the longest edge in $M$. We denote its length by $\bt_M$. A {\em bottleneck perfect matching} is a perfect matching that minimizes the bottleneck length. A {\em plane matching} is a matching with non-crossing edges. 
We denote a plane matching by $M_=$ and a crossing matching by $M_\times$. 
The {\em bottleneck plane perfect matching}, $\MOPT$, is a perfect plane matching which minimizes the length of the longest edge. Let $\btopt$ denote the length of the bottleneck edge in $\MOPT$. In this paper we consider the problem of computing a bottleneck plane matching of $P$. 

The Unit Disk Graph, $UDG(P)$, is defined to have the points of $P$ as its vertices and two vertices $p$ and $q$ are connected by an edge if their Euclidean distance $|pq|$ is at most 1. The {\em maximum plane matching} of $UDG(P)$ is the maximum cardinality matching of $UDG(P)$, which has no pair of crossing edges. 

\begin{lemma}
 If the maximum plane matching in unit disk graphs can be computed in polynomial time, then the bottleneck plane perfect matching problem for point sets can also be solved in polynomial time.
\end{lemma}
\begin{proof}
Let $D=\{|pq|:p,q\in P\}$ be the set of all distances determined by pairs of points in $P$. Note that $\btopt\in D$. For each $\bt\in D$, define the ``unit'' disk graph $DG(\bt,P)$, in which two points $p$ and $q$ are connected by an edge if and only if $|pq|\le\bt$. Then $\btopt$ is the minimum $\bt$ in $D$ such that $DG(\bt,P)$ has a plane matching of size $\frac{n}{2}$.
\end{proof}

The Gabriel Graph, $GG(P)$, is defined to have the points of $P$ as its vertices and two vertices $p$ and $q$ are connected by an edge if the disk with diameter $pq$ does not contain any point of $P\setminus\{p,q\}$ in its interior and on its boundary.

\begin{lemma}
\label{MST-UDG}
If the unit disk graph $UDG(P)$ of a point set $P$ is connected, then $UDG(P)$ and $K(P)$ have the same minimum spanning tree.
\end{lemma}
\begin{proof}
 By running Kruskal's algorithm on $UDG(P)$, we get a minimum spanning tree, say $T$. All the edges of $T$ have length at most one, and the edges of $K(P)$ which do not belong to $UDG(P)$ all have length greater than one. Hence, $T$ is also a minimum spanning tree of $K(P)$.
\end{proof}

As a direct consequence of Lemma \ref{MST-UDG} we have the following corollary:

\begin{corollary}
\label{MST-DEL}
Consider the unit disk graph $UDG(P)$ of a point set $P$. We can compute the minimum spanning forest of $UDG(P)$, by first computing the minimum spanning tree of $P$ and then removing the edges whose length is more than one.
\end{corollary}

\begin{lemma}
\label{same-component}
 For each pair of crossing edges $(u,v)$ and $(x,y)$ in $UDG(P)$, the four endpoints $u$, $v$, $x$, and $y$ are in the same component of $UDG(P)$. 
\end{lemma}
\begin{proof}
Note that the quadrilateral $Q$ formed by the end points $u$, $v$, $x$, and $y$ is convex. W.l.o.g. assume that the angle $\angle xuy$ is the largest angle in $Q$. Clearly $\angle xuy \ge \pi/2$, and hence, in triangle $\bigtriangleup xuy$, the angles $\angle yxu$ and $\angle uyx$ are both less than $\pi/2$. Thus, the edges $(u,x)$ and $(u,y)$ are both less than $(x,y)$. This means that $(u,x)$ and $(u,y)$ are also edges of $UDG(P)$, thus, their four endpoints belong to the same component.  
\end{proof}

As a direct consequence of Lemma \ref{same-component} we have the following corollary: 

\begin{corollary}
\label{non-crossing-edges}
Any two edges that belong to different components of $UDG(P)$ do not cross.
\end{corollary}

\begin{figure}[ht]
  \centering
    \includegraphics[width=0.4\textwidth]{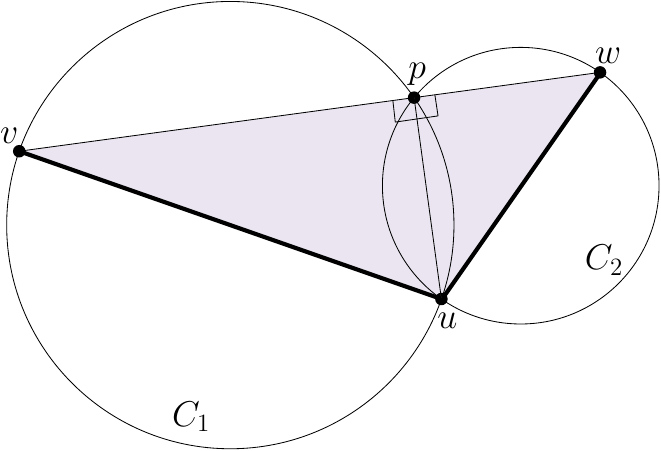}
  \caption{In $MST(P)$, the triangle $\bigtriangleup uvw$ formed by two adjacent edges $(u,v)$ and $(u,w)$, is empty.}
\label{empty-triangle-figure}
\end{figure}

Let $MST(P)$ denote the Euclidean minimum spanning tree of $P$.

\begin{lemma}
\label{empty-triangle-lemma}
If $(u,v)$ and $(u,w)$ are two adjacent edges in $MST(P)$, then the triangle $\bigtriangleup uvw$ has no point of $P\setminus\{u, v, w\}$ inside or on its boundary.
\end{lemma}
\begin{proof}
If the angle between the line segments $\overline{uv}$ and $\overline{uw}$ is equal to $\pi$ then clearly there is no other point of $P$ on $\overline{uv}$ and $\overline{uw}$. So assume that $\angle vuw < \pi$. Refer to Figure \ref{empty-triangle-figure}. Since $MST(P)$ is a subgraph of the Gabriel graph, the circles $C_1$ and $C_2$ with diameters $\overline{uv}$ and $\overline{uw}$ are empty. Since $\angle vuw < \pi$, $C_1$ and $C_2$ intersect each other at two points, say $u$ and $p$. Connect $u$, $v$ and $w$ to $p$. Since $\overline{uv}$ and $\overline{uw}$ are the diameters of $C_1$ and $C_2$, $\angle upv=\angle wpu=\pi/2$.
This means that $\overline{vw}$ is a straight line segment. Since $C_1$ and $C_2$ are empty and $\bigtriangleup uvw \subset C_1 \cup C_2$, it follows that $\bigtriangleup uvw \cap P = \{u, v, w\}$.
\end{proof}

\begin{corollary}
\label{empty-convex-hull}
Consider $MST(P)$ and let $N_v$ be the set of neighbors of a vertex $v$ in $MST(P)$. Then the convex hull of $N_v$ contains no point of $P$ except $v$ and the set $N_v$.
\end{corollary}

The shaded area in Figure \ref{empty-skeleton-fig} shows the union of all these convex hulls.

\section{Plane Matching in Unit Disk Graphs}
\label{UDG}

In this section we present two approximation algorithms for computing a maximum plane matching in a unit disk graph $UDG(P)$. In Section \ref{three-approximation} we present a straight-forward $\frac{1}{3}$-approximation algorithm; it is unclear whether this algorithm runs in polynomial time. For the case when $UDG(P)$ is connected, we present a $\frac{2}{5}$-approximation algorithm in Section \ref{five-over-two-approximation}.

\subsection{$\frac{1}{3}$-approximation algorithm}
\label{three-approximation}

Given a possibly disconnected unit disk graph $UDG(P)$, we start by computing a (possibly crossing) maximum matching $\MC$ of $UDG(P)$ using Edmonds algorithm \cite{Edmonds1965}. Then we transform $\MC$ to another (possibly crossing) matching $\MC'$ with some properties, and then pick at least one-third of its edges which satisfy the non-crossing property. Consider a pair of crossing edges $(p,q)$ and $(r,s)$ in $\MC$, and let $c$ denote the intersection point. If their smallest intersection angle is at most $\pi/3$, we replace these two edges with new ones. For example if $\angle pcr \le \pi/3$, we replace $(p,q)$ and $(r,s)$ with new edges $(p,r)$ and $(q,s)$. Since the angle between them is at most $\pi/3$, the new edges are not longer than the older ones, i.e. $\max\{|pr|,|qs|\}\le\max\{|pq|,|rs|\}$, and hence the new edges belong to the unit disk graph. On the other hand the total length of the new edges is strictly less than the older ones; i.e. $|pr|+|qs|<|pq|+|rs|$. For each pair of intersecting edges in $\MC$, with angle at most $\pi/3$, we apply this replacement. We continue this process until we have a matching $\MC'$ with the property that if two matched edges intersect, each of the angles incident on $c$ is larger than $\pi/3$.

For each edge in $M'_{\times}$, consider the counter clockwise angle it makes with the positive $x$-axis; this angle is in the range $[0,\pi)$. Using these angles, we partition the edges of $M'_{\times}$ into three subsets, one subset for the angles $[0,\pi/3)$, one subset for the angles $[\pi/3,2\pi/3)$, and
one subset for the angles $[2\pi/3,\pi)$. Observe that edges within one subset are non-crossing. Therefore, if we output the largest subset, we obtain a non-crossing matching of size at least
$|M'_{\times}|/3 = |M_{\times}|/3$.

Since in each step (each replacement) the total length of the matched edges decreases, the replacement process converges and the algorithm will stop. We do not know whether this process converges in a polynomial number of steps in the size of $UDG(P)$.

\subsection{$\frac{2}{5}$-approximation algorithm for connected unit disk graphs}
\label{five-over-two-approximation}

In this section we assume that the unit disk graph $UDG(P)$ is connected. Monma et al. \cite{Monma1992} proved that every set of points in the plane admits a minimum spanning tree of degree at most five which can be computed in $O(n\log n)$ time. By Lemma \ref{MST-UDG}, the same claim holds for $UDG(P)$. Here we present an algorithm which extracts a plane matching $M$ from $MST(P)$. Consider a minimum spanning tree $T$ of $UDG(P)$ with vertices of degree at most five. We define the {\em skeleton tree}, ${T'}$, as the tree obtained from $T$ by removing all its leaves; see Figure \ref{empty-skeleton-fig}. Clearly ${T'} \subseteq T \subseteq UDG(P)$. For clarity we use $u$ and $v$ to refer to the leaves of $T$ and $T'$ respectively. In addition, let $v$ and $v'$, respectively, refer to the copies of a vertex $v$ in $T$ and $T'$. In each step, pick an arbitrary leaf $v'\in T'$. By the definition of ${T'}$, it is clear that the copy of $v'$ in $T$, i.e. $v$, is connected to vertices $u_1,\dots, u_k$, for some $1\le k \le 4$, which are leaves of $T$ (if $T'$ has one vertex then $k\le5$). Pick an arbitrary leaf $u_i$ and add $(v, u_i)$ as a matched pair to $M$. For the next step we update $T$ by removing $v$ and all its adjacent leaves. We also compute the new skeleton tree and repeat this process. 
In the last iteration, $T'$ is empty and we may be left with a tree $T$ consisting of one single vertex or one single edge. If $T$ consists of one single vertex, we disregard it, otherwise we add its only edge to $M$.

\begin{figure}[ht]
  \centering
    \includegraphics[width=0.7\textwidth]{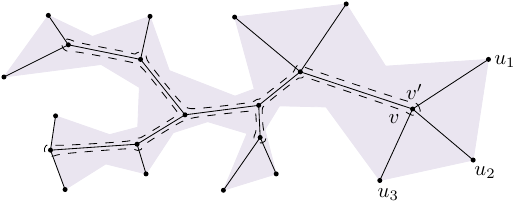}
  \caption{Minimum spanning tree $T$ with union of empty convex hulls. The skeleton tree $T'$ is surrounded by dashed line, and $v'$ is a leaf in $T'$. }
\label{empty-skeleton-fig}
\end{figure}

The formal algorithm is given as {\scshape PlaneMatching}, which receives a point set $P$---whose unit disk graph is connected---as input and returns a matching $M$ as output. The function MST5$(P)$ returns a Euclidean minimum spanning tree of $P$ with degree at most five, and the function Neighbor$(v', T')$ returns the only neighbor of leaf $v'$ in $T'$.

\begin{algorithm}                      
\caption{{\scshape PlaneMatching}$(P)$}          
\label{alg1} 
\require{set $P$ of $n$ points in the plane, such that $UDG(P)$ is connected.}\\
\ensure{plane matching $M$ of $MST(P)$ with $|M|\ge\frac{n-1}{5}$.}
\begin{algorithmic}[1]
    \State $M \gets \emptyset$
    \State $T \gets $MST5$(P)$
    \State $T' \gets $SkeletonTree$(T)$
    \While {$T' \neq \emptyset$}
	\State $v' \gets$ {a leaf of } $T'$
	\State $L_v \gets$ {set of leaves  connected to } $v$ { in } $T$
        \State $u \gets$ {an element of } $L_v$
	\State $M \gets M\cup\{(v,u)\}$	
	\State $T\gets T\setminus (\{v\}\cup L_v)$
	\If {$deg($Neighbor$(v', T'))=2$}
	    \State $T'\gets T'\setminus \{v',${ Neighbor}$(v', T')\}$
	\Else
	    \State $T'\gets T'\setminus \{v'\}$
	\EndIf
    \EndWhile
    \If {$T$ {consists of one single edge}}
	\State $M \gets M\cup T$
    \EndIf
    \State \Return $M$
\end{algorithmic}
\end{algorithm}

\begin{lemma}
\label{n-minus-one-lemma}
Given a set $P$ of $n$ points in the plane such that $UDG(P)$ is connected, algorithm {\scshape PlaneMatching} returns a plane matching $M$ of $MST(P)$ of size $|M|\ge \frac{n-1}{5}$. Furthermore, $M$ can be computed in $O(n\log n)$ time.
\end{lemma}
\begin{proof}
In each iteration an edge $(v, u_i)\in T$ is added to $M$. Since $T$ is plane, $M$ is also plane and $M$ is a matching of $MST(P)$.

Line 5 picks $v'\in T'$ which is a leaf, so its analogous vertex $v\in T$ is connected to at least one leaf. In each iteration we select an edge incident to one of the leaves and add it to $M$, then disregard all other edges connected to $v$ (line 9). So for the next iteration $T$ looses at most five edges. Since $T$ has $n-1$ edges initially and we add one edge to $M$ out of each five edges of $T$, we have $|M|\ge\frac{n-1}{5}$. 

According to Corollary \ref{MST-DEL} and by \cite{Monma1992}, line 2 takes $O(n\log n)$ time. The while-loop is iterated $O(n)$ times and in each iteration, lines 5-13 take constant time. So, the total running time of algorithm {\scshape PlaneMatching} is $O(n\log n)$. 
\end{proof}

The size of a maximum matching can be at most $\frac{n}{2}$. Therefore, algorithm {\scshape PlaneMatching} computes a matching of size at least $\frac{2(n-1)}{5n}$ times the size of a perfect matching, and hence, when $n$ is large enough, {\scshape PlaneMatching} is a $\frac{2}{5}$-approximation algorithm. On the other hand there are unit disk graphs whose maximum matchings have size $\frac{n-1}{5}$; see Figure \ref{worst-case-optimal}. In this case {\scshape PlaneMatching} returns a maximum matching. In addition, when $UDG(P)$ is a tree or a cycle, {\scshape PlaneMatching} returns a maximum matching. 

\begin{figure}[ht]
  \centering
    \includegraphics[width=0.7\textwidth]{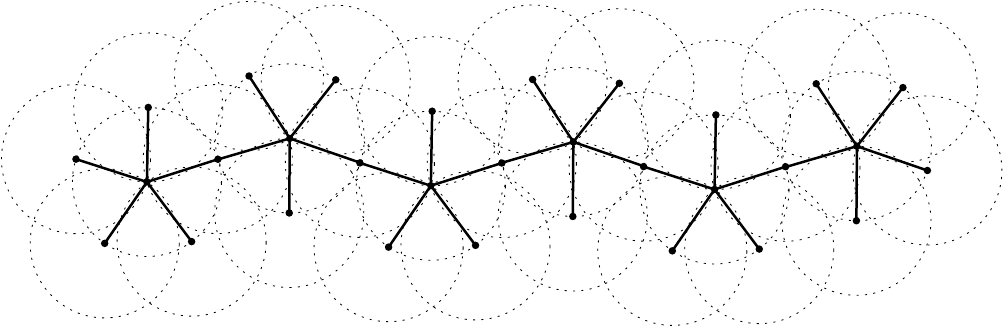}
  \caption{Unit disk graph with all edges of unit length, and maximum matching of size $\frac{n-1}{5}$.}
\label{worst-case-optimal}
\end{figure}

In Section \ref{bottleneck-five-over-two} we will show how one can use a modified version of algorithm {\scshape PlaneMatching} to compute a bottleneck plane matching of size at least $\frac{n}{5}$ with bottleneck length at most $\btopt$. Recall that $\btopt$ is the length of the bottleneck edge in the bottleneck plane perfect matching $\MOPT$. Section \ref{bottleneck-five-over-four} extends this idea to an algorithm which computes a plane matching of size at least $\frac{2n}{5}$ with edges of length at most $(\sqrt{2}+\sqrt{3})\btopt$.

\section{Approximating Bottleneck Plane Perfect Matching}
\label{bottleneck}

The general approach of our algorithms is to first compute a (possibly crossing) bottleneck perfect matching $\MC$ of $K(P)$ using the algorithm in \cite{Chang1992}. Let $\bt_{\MC}$ denote the length of the bottleneck edge in $\MC$. It is obvious that the bottleneck length of any plane perfect matching is not less than $\bt_{\MC}$. Therefore, $\btopt\ge\bt_{\MC}$. We consider a ``unit'' disk graph $DG(\bt_{\MC}, P)$ over $P$, in which there is an edge between two vertices $p$ and $q$ if $|pq|\le\bt_{\MC}$. Note that $DG(\bt_{\MC}, P)$ is not necessarily connected. Let $G_1,\dots,G_k$ be the connected components of $DG(\bt_{\MC}, P)$. For each component $G_i$, consider a minimum spanning tree $T_i$ of degree at most five. We show how to extract from $T_i$ a plane matching $M_i$ of proper size and appropriate edge lengths.

\begin{lemma}
 Each component of $DG(\bt_{\MC}, P)$ has an even number of vertices.
\end{lemma}
\begin{proof}
This follows from the facts that $\MC$ is a perfect matching and both end points of each edge in ${\MC}$ belong to the same component of $DG(\bt_{\MC}, P)$.
\end{proof}

\subsection{First Approximation Algorithm}
\label{bottleneck-five-over-two}
In this section we describe the process of computing a plane matching $M$ of size
$|M|\ge\frac{1}{5}n$ with bottleneck length at most $\btopt$. Consider the minimum spanning trees $T_1,\dots,T_k$ of the $k$ components of $DG(\bt_{\MC}, P)$. For $1\le i\le k$, let $P_i$ denote the set of vertices in $T_i$ and $n_i$ denote the number of vertices of $P_i$. Our approximation algorithm runs in two steps:

\begin{paragraph}{Step 1:} 
We start by running algorithm {\scshape PlaneMatching} on each of the point sets $P_i$. Let $M_i$ be the output. Recall that algorithm {\scshape PlaneMatching}, from Section \ref{five-over-two-approximation}, picks a leaf $v'\in T_i'$, corresponding to a vertex $v \in T_i$, matches it to one of its neighboring leaves in $T_i$ and disregards the other edges connected to $v$. According to Lemma \ref{n-minus-one-lemma}, this gives us a plane matching $M_i$ of size at least $\frac{n_i-1}{5}$. However, we are looking for a matching of size at least $\frac{n_i}{5}$. 

 The total number of edges of $T_i$ is $n_i-1$ and in each of the iterations, the algorithm picks one edge out of at most five candidates. If in at least one of the iterations of the while-loop, $v$ has degree at most four (in $T_i$), then in that iteration algorithm {\scshape PlaneMatching} picks one edge out of at most four candidates. Therefore, the size of $M_i$ satisfies 
$$
 |M_i| \ge 1 + \frac{(n_i-1)-4}{5}=\frac{n_i}{5}.
$$
 If in all the iterations of the while-loop, $v$ has degree five, we look at the angles between the consecutive leaves connected to $v$. Recall that in $MST(P)$ all the angles are greater than or equal to $\pi/3$. If in at least one of the iterations, $v$ is connected to two consecutive leaves $u_j$ and $u_{j+1}$ for $1\le j\le 3$, such that $\angle u_jvu_{j+1} = \pi/3$, we change $M_i$ as follow. Remove from $M_i$ the edge incident to $v$ and add to $M_i$ the edges $(u_j,u_{j+1})$ and $(v,u_s)$, where $u_s, s\notin\{j, j+1\}$, is one of the leaves connected to $v$. Clearly $\bigtriangleup vu_ju_{j+1}$ is equilateral and $|u_ju_{j+1}| = |u_jv|=|u_{j+1}v|\le \bt_{\MC}$, and by Lemma \ref{empty-triangle-lemma}, $(u_j, u_{j+1})$ does not cross other edges. In this case, the size of $M_i$ satisfies 
$$
 |M_i| = 2+\frac{(n_i-1)-5}{5}=\frac{n_i+4}{5}\ge\frac{n_i}{5}.
$$

\end{paragraph}

\begin{paragraph}{Step 2:}
In this step we deal with the case that in all the iterations of the while-loop, $v$ has degree five and the angle between any pair of consecutive leaves connected to $v$ is greater than $\pi/3$. Recall that $\MC$ is a perfect matching and both end points of each
edge in $\MC$ belong to the same $T_i$.

\begin{figure}[ht]
  \centering
    \includegraphics[width=0.5\textwidth]{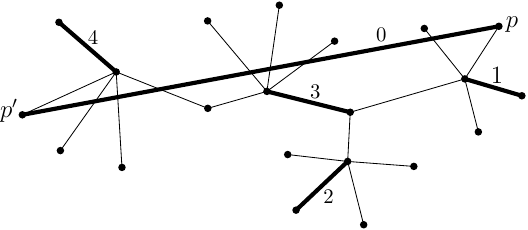}
  \caption{Resulted matching $M_i$ by modified {\scshape PlaneMatching}. The numbers show the order in which the edges (bold edges) are added to $M_i$.}
\label{two-leaves-fig}
\end{figure}

\begin{lemma}
\label{two-leaves-lemma}
In Step 2, at least two leaves of $T_i$ are matched in $\MC$.
\end{lemma}
\begin{proof}
 Let $m_i$ and $m_i'$ denote the number of external (leaves) and internal nodes of $T_i$, respectively. Clearly $m_i'$ is equal to the number of vertices of $T_i'$ and $m_i+m_i'=n_i$. Consider the reverse process in {\scshape PlaneMatching}. Start with a 5-star tree $t_i$, i.e. $t_i=K_{1,5}$, and in each iteration append a new $K_{1,5}$ to $t_i$ until $t_i=T_i$. In the first step $m_i=5$ and $m_i'=1$. In each iteration a leaf of the appended $K_{1,5}$ is identified with a leaf of $t_i$; the resulting vertex becomes an internal node. On the other hand, the \textquotedblleft center\textquotedblright of the new $K_{1,5}$ becomes an internal node of $t_i$ and its other four neighbors become leaves of $t_i$. So in each iteration, the number of leaves $m_i$ increases by three, and the number of internal nodes $m_i'$ increases by two. Hence, in all iterations (including the first step) we have $m_i\ge m'_i+4$.

Again consider $\MC$. In the worst case if all $m_i'$ internal vertices of $T_i$ are matched to leaves, we still have four leaves which have to be matched together.   
\end{proof}

According to Lemma \ref{two-leaves-lemma} there is an edge $(p,p')\in \MC$ where $p$ and $p'$ are leaves in $T_i$. We can find $p$ and $p'$ for all $T_i$'s by checking all the edges of $\MC$ once. We remove all the edges of $M_i$ and initialize $M=\{(p,p')\}$. Again we run a modified version of {\scshape PlaneMatching} in such a way that in each iteration, in line 7 it selects the leaf $u_i$ adjacent to $v$ such that $(v, u_i)$ is not intersected by $(p,p')$. In each iteration $v$ has degree five and is connected to at least four leaf edges with angles greater than $\pi/3$. Thus, $(p,p')$ can intersect at most three of the leaf edges and such kind of $(v, u_i)$ exists. See Figure \ref{two-leaves-fig}. In this case, $M_i$ has size
$$
 |M_i| = 1+\frac{n_i-1}{5}=\frac{n_i+4}{5}\ge\frac{n_i}{5}.
$$

\end{paragraph}

We run the above algorithm on each $T_i$ and for each of them we compute a plane matching $M_i$. The final matching of point set $P$ will be $M=\bigcup_{i=1}^{k} M_i$.

\begin{theorem}
\label{two-over-five-theorem}
Let $P$ be a set of $n$ points in the plane, where $n$ is even, and let $\btopt$ be the minimum bottleneck length of any plane perfect matching of $P$. In $O(n^{1.5}\sqrt{\log n})$ time, a plane matching of $P$ of size at least $\frac{n}{5}$ can be computed, whose bottleneck length is at most $\btopt$.
\end{theorem}
\begin{proof}

{\em Proof of edge length}: Let $\bt_{\MC}$ be the length of the longest edge in $\MC$ and consider a component $G_i$ of $DG(\bt_{\MC}, P)$. All the selected edges in Steps 1 and 2 belong to $T_i$ except $(u_j,u_{j+1})$ and $(p,p')$. $T_i$ is a subgraph of $G_i$, and the edge $(u_j,u_{j+1})$ belongs to $G_i$, and the edge $(p,p')$ belongs to $\MC$ (which belongs to $G_i$ as well). So all the selected edges belong to $G_i$, and $\bt_{M_i}\le \bt_{\MC}$. Since $\bt_{\MC}\le \btopt$, we have $\bt_{M_i}\le \bt_{\MC}\le \btopt$ for all $i$, $1\le i \le k$.

 {\em Proof of planarity}: The edges of $M_i$ belong to the minimum spanning forest of $DG(\bt_{\MC}, P)$ which is plane, except $(u_j,u_{j+1})$ and $(p,p')$. According to Corollary \ref{non-crossing-edges} and Lemma \ref{empty-triangle-lemma} the edge $(u_j,u_{j+1})$ does not cross the edges of the minimum spanning forest. In Step 2 we select edges of $T_i$ in such a way that avoid $(p,p')$. Note that $(p,p')$ belongs to the component $G_i$ and by Corollary \ref{non-crossing-edges} it does not cross any edge of the other components of $DG(\bt_{\MC}, P)$. So $M$ is plane.

 {\em Proof of matching size}: Since $M=M_1\cup \dots \cup M_k$, and for each $1\le i\le k$, $|M_i|\ge \frac{n_i}{5}$, hence
$$
 |M|\ge \sum_{i=1}^{k} |M_i| \ge \sum_{i=1}^{k} \frac{n_i}{5} = \frac{n}{5}.
$$

 {\em Proof of complexity}: The initial matching $\MC$ can be computed in time $O(n^{1.5}\sqrt{\log n})$ by using the algorithm of Chang et al. \cite{Chang1992}. By Lemma \ref{n-minus-one-lemma} algorithm {\scshape PlaneMatching} runs in $O(n \log n)$ time. In Step 1 we spend constant time for checking the angles and the number of leaves connected to $v$ during the while-loop. In Step 2, the matched leaves $p$ and $p'$ can be computed in $O(n)$ time for all $T_i$'s by checking all the edges of $\MC$ before running the algorithm again. So the modified {\scshape PlaneMatching} still runs in $O(n \log n)$ time, and the total running time of our method is $O(n^{1.5}\sqrt{\log n})$.
\end{proof}

Since the running time of the algorithm is bounded by the time of computing the initial bottleneck matching $\MC$, any improvement in computing $\MC$ leads to a faster algorithm for computing a plane matching  $M$. In the next section we improve the running time.

\subsubsection{Improving the Running Time}
In this section we present an algorithm that improves the running time to $O(n\log^2 n)$. We first compute a forest $F$, such that each edge in $F$ is of length at most $\btopt$ and, for each tree $T \in F$, we have a leaf $p \in T$ and a point $p' \in T$ such that $|pp'| \le \btopt$. Once we have this forest $F$, we apply Step 1 and Step 2 on $F$ to obtain the matching $M$ as in the previous section. 

Let $MST(P)$ be a (five-degree) minimum spanning tree of $P$. 
Let $F_{\lambda}=\{ T_1, T_2, \dots, T_k\}$ be the forest obtained from $MST(P)$ by removing all the edges whose length is greater than $\lambda$, i.e., $F_{\lambda} = \{e \in MST(P): |e| \le \lambda\}$. For a point $p \in P$, let $cl(p, P)$ be the point in $P$ that is closest to $p$.

\begin{lemma}
\label{opt-forest}
For all $T \in F_{\btopt}$, it holds that
\begin{enumerate}
	\item[(i)]  the number of points in $T$ is even, and
	\item[(ii)]  for each two leaves $p,q \in T$ that are incident to the same node $v$ in $T$, let $p' = cl(p, P\setminus \{v\})$, let $q' = cl(q, P\setminus \{v\})$, and assume that $|pp'| \le |qq'|$. Then, $\btopt \geq |pp'|$ and $p'$ belongs to $T$.
\end{enumerate}
\end{lemma}
\begin{proof}
(i) Suppose that $T$ has odd number of points. Thus in $M^*$ one of the points in $T$ should be matched to a point in a tree $T' \neq T$ by an edge $e$. Since $e \notin F_{\btopt}$, we have $|e|>\btopt$, which contradicts that $\btopt$ is the minimum bottleneck. (ii) Note that $v$ is the closest point to both $p$ and $q$. In $M^*$, at most one of $p$ and $q$ can be matched to $v$, and the other one must be matched to a point which is at least as far as its second closest point. Thus, $\btopt$ is at least $|pp'|$. The distance between any two trees in $F_{\btopt}$ is greater than $\btopt$. Now if $p'$ is not in $T$, then in any bottleneck perfect matching, either $p$ or $q$ is matched to a point of distance greater than $\btopt$, which contradicts that $\btopt$ is the minimum bottleneck.
\end{proof}

Let $E=(e_1, e_2, \dots ,e_{n-1})$ be the edges of $MST(P)$ in sorted order of their lengths. Our algorithm performs a binary search on $E$, and for each considered edge $e_i$, it uses Algorithm~\ref{proc2} to decide whether $\lambda < \btopt$, where $\lambda=|e_i|$. 
The algorithm constructs the forest $F_{\lambda}$, and for each tree $T$ in $F_{\lambda}$, it picks two leaves $p$ and $q$ from $T$ and finds their second closest points $p'$ and $q'$. Assume w.l.o.g. that $|pp'| \le |qq'|$. Then, the algorithm returns FALSE if $p'$ does not belong to $T$. By Lemma~\ref{opt-forest}, if the algorithm returns FALSE, then we know that $\lambda <\btopt$. 

Let $e_j$ be the shortest edge in $MST(P)$, for which Algorithm~\ref{proc2} does not return FALSE. This means that Algorithm~\ref{proc2} returns FALSE for $e_{j-1}$, and there is a tree $T$ in $F_{|e_{j-1}|}$ and a leaf $p$ in $T$, such that $|pp'|\ge |e_j|$. Thus $|e_j|\leq\btopt$ and for each tree $T$ in the forest $F_{|e_j|}$, we stored a leaf $p$ of $T$ and a point $p' \in T$, such that $|pp'| \le \btopt$. Since each tree in $F_{|e_j|}$ is a subtree of $MST(P)$, $F_{|e_j|}$ is planar and each tree in $F_{|e_j|}$ is of degree at most five. 

\floatname{algorithm}{Algorithm}

\begin{algorithm}                      
\caption{{\scshape CompareToOpt}$(\lambda)$}          
\label{proc2} 
\begin{algorithmic}[1]

  \State compute $F_{\lambda}$ 
  \State {$L\gets$ empty list}
  \For {each $T \in F_{\lambda}$}
	\If {$T$ has an odd number of points} 
		\State return {FALSE} \EndIf
	\If {there exist two leaves $p$ and $q$ incident to a node $v\in T$}
  \State $p' \gets cl(p, P\setminus \{v\} )$
	\State $q' \gets cl(q, P\setminus \{v\} )$
	
  \If {$|pp'| \le |qq'|$}
			\If {$p'$ does not belong to $T$} 
				\State \Return {FALSE}
			\Else
				\State {add the triple $(p,p',T)$ to $L$}  \EndIf
  \Else 
			\If {$q'$ does not belong to $T$} 
				\State \Return {FALSE}
			\Else
				\State {add the triple $(q,q',T)$ to $L$}  \EndIf
  \EndIf \EndIf
	\EndFor
\State \Return $L$
	
\end{algorithmic}
\end{algorithm}

Now we can apply Step 1 and Step 2 on $F_{|e_j|}$ as in the previous section. Note that in Step 2, for each tree $T_i$ we have a pair $(p,p')$ (or $(q,q')$) in the list $L$ which can be matched. In Step 2, in each iteration $v$ has degree five, thus, $p'$ should be a vertex
of degree two or a leaf in $T_i$. If $p'$ is a leaf we run the modified version of {\scshape PlaneMatching} as in the previous section. If $p'$ has degree two, we remove all the edges of $M_i$ and initialize $M_i = {(p, p')}$. Then remove $p'$ from $T_i$ and run {\scshape PlaneMatching} on the resulted subtrees. Finally, set $M=\bigcup_{i=1}^{k}M_i$.

\begin{lemma}
\label{planarity}
 The matching $M$ is planar.
\end{lemma}
\begin{proof}
Consider two edges $e=(p,p')$ and $e'=(q,q')$ in $M$. We distinguish between four cases:
\begin{enumerate}
   \item $e\in F_{|e_j|}$ and $e' \in F_{|e_j|}$. In this case, both $e$ and $e'$ belong to $MST(P)$ and hence they do not cross each other.
  \item $e\notin F_{|e_j|}$ and $e'\in F_{|e_j|}$. If $e$ and $e'$ cross each other, then this contradicts the selection of $(q,q')$ in Step 2 (which prevents $(p,p')$).
  \item $e\in F_{|e_j|}$ and $e'\notin F_{|e_j|}$. It leads to a contradiction as in the previous case.
  \item $e\notin F_{|e_j|}$ and $e'\notin F_{|e_j|}$. If $e$ and $e'$ cross each other, then either $\min\{|pq|,|pq'|\} < |pp'|$ or $\min\{|qp|,|qp'|\} < |qq'|$, which contradicts the selection of $p'$ or $q'$. Note that $p$ cannot be the second closest point to $q$, because $p$ and $q$ are in different trees. 
\end{enumerate} 
\end{proof}

\begin{lemma}
\label{running-time}
 The matching $M$ can be computed in $O(n\log^2 n)$ time.
\end{lemma}
\begin{proof}
Computing $MST(P)$ and sorting its edges take $O(n\log{n})$~\cite{deBerg08}. Since we performed a binary search on the edges of $MST(P)$, we need $\log{n}$ iterations. In each iteration, for an edge $e_i$, we compute the forest $F_{|e_i|}$ in $O(n)$ and the number of the trees in the forest can be $O(n)$ in the worst case. We compute in advance the second order Voronoi diagram of the points together with a corresponding point location data structure, in $O(n\log{n})$~\cite{deBerg08}. For each tree in the forest, we perform a point location query to find the closest points $p'$ and $q'$, which takes $O(\log{n})$ for each query. Therefore the total running time is $O(n\log^2{n})$.
\end{proof}

\begin{theorem}
\label{two-over-five-theorem-2}
Let $P$ be a set of $n$ points in the plane, where $n$ is even, and let $\btopt$ be the minimum bottleneck length of any plane perfect matching of $P$. In $O(n\log^2 n)$ time, a plane matching of $P$ of size at least $\frac{n}{5}$ can be computed, whose bottleneck length is at most $\btopt$.
\end{theorem}

\subsection{Second Approximation Algorithm}
\label{bottleneck-five-over-four}

In this section we present another approximation algorithm which gives a plane matching $M$ of size $|M|\ge\frac{2}{5}n$ with bottleneck length $\bt_M\le(\sqrt{2}+\sqrt{3})\btopt$. Let $DT(P)$ denote the Delaunay triangulation of $P$. Let the edges of $DT(P)$ be, in sorted order of their lengths, $e_1, e_2, \dots$. Initialize a forest $F$ consisting of $n$ tress, each one being a single node for one point of $P$. Run Kruskal's algorithm on the edges of $DT(P)$ and terminate as soon as every tree in $F$ has an even number of nodes. Let $e_l$ be the last edge that is added by Kruskal's algorithm. Observe that $e_l$ is the longest edge in $F$. Denote the trees in $F$ by $T_1 ,\dots, T_k$ and for $1\le i \le k$, let $P_i$ be the vertex set of $T_i$ and let $n_i=|P_i|$. 

\begin{lemma}
\label{longest-edge}
 $\btopt \ge |e_l|.$
\end{lemma}
\begin{proof}
Let $i$ be such that $e_l$ is an edge in $T_i$. Let $T'_i$ and $T''_i$ be the trees obtained by removing $e_l$ from $T_i$. Let $P'_i$ be the vertex set of $T'_i$. Then $|e_l|=\min \{|pq|: p \in P'_i, q\in P\setminus P'_i\}$. Consider the optimal matching $M^*$ with bottleneck length $\btopt$. Since $e_l$ is the last edge added, $P'_i$ has odd size. The matching $M^*$ contains an edge joining a point in $P'_i$ with a point in $P\setminus P'_i$. This edge has length at least $|e_l|$.
\end{proof}

By Lemma \ref{longest-edge} the length of the longest edge in $F$ is at most $\btopt$. For each $T_i\in F$, where $1\le i\le k$, our algorithm will compute a plane matching $M_i$ of $P_i$ of size at least $\frac{2n_i}{5}$ with edges of length at most $(\sqrt{2}+\sqrt{3})\btopt$ and returns $\bigcup_{i=1}^{k} M_i$. To describe the algorithm for tree $T_i$ on vertex set $P_i$, we will write $P$, $T$, $n$, $M$ instead of $P_i$, $T_i$, $n_i$, $M_i$, respectively. Thus, $P$ is a set of $n$ points, where $n$ is even, and $T$ is a minimum spanning tree of $P_i$.

Consider the minimum spanning tree $T$ of $P$ having degree at most five, and let ${T'}$ be the skeleton tree of $T$. Suppose that $T'$ has at least two vertices. We will use the following notation. Let $v'$ be a leaf in $T'$, and let $w'$ be the neighbor of $v'$. Recall that $v'$ and $w'$ are copies of vertices $v$ and $w$ in $T$. In $T$, we consider the clockwise ordering of the neighbors of $v$. Let this ordering be $w, u_1, u_2, \dots, u_k$ for some $1\le k\le 4$. Clearly $u_1,\dots,u_k$ are leaves in $T$. Consider two leaves $u_i$ and $u_j$ where $i<j$. We define $\cw(u_ivu_j)$ as the clockwise angle from $\overline{vu_i}$ to $\overline{vu_j}$. We say that the leaf $v'$ (or its copy $v$) is an {\em anchor} if $k=2$ and $\cw(u_1vu_2)\ge\pi$. See Figure \ref{empty-skeleton-angle-fig}.

\begin{figure}[ht]
  \centering
    \includegraphics[width=0.7\textwidth]{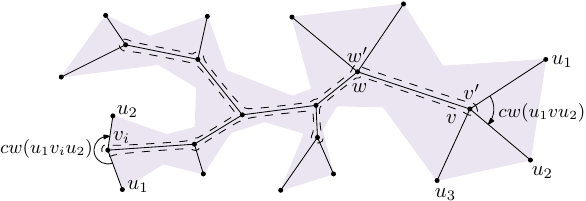}
  \caption{The vertices around $v$ are sorted clockwise, and $v_i$ is an anchor vertex.}
\label{empty-skeleton-angle-fig}
\end{figure}

Now we describe how one can iteratively compute a plane matching of proper size with bounded-length edges from $T$. We start with an empty matching $M$. Each iteration consists of two steps, during which we add edges to $M$. As we prove later, the output is a plane matching of $P$ of size at least $\frac{2}{5}n$ with bottleneck at most $(\sqrt{2}+\sqrt{3})\btopt$.

\subsubsection{Step 1}

We keep applying the following process as long as $T$ has more than six vertices and $T'$ has some non-anchor leaf. Note that $T'$ has at least two vertices.
 
Take a non-anchor leaf $v'$ in $T'$ and according to the number $k$ of leaves connected to $v$ in $T$ do the following:
 \begin{description}
  \item[{\em k} = 1] add $(v, u_1)$ to $M$, and set $T=T\setminus\{v, u_1\}$. 
  \item[{\em k} = 2] since $v'$ is not an anchor, $\cw(u_1vu_2)< \pi$. By Lemma \ref{empty-triangle-lemma} the triangle $\bigtriangleup u_1vu_2$ is empty. We add $(u_1, u_2)$ to $M$, and set $T=T\setminus\{u_1, u_2\}$.
  \item[{\em k} = 3] in this case $v$ has degree four and at least one of $\cw(u_1vu_2)$ and $\cw(u_2vu_3)$ is less than $\pi$. W.l.o.g. suppose that $\cw(u_1vu_2)< \pi$. By Lemma \ref{empty-triangle-lemma} the triangle $\bigtriangleup u_1vu_2$ is empty. Add $(u_1, u_2)$ to $M$ and set $T=T\setminus\{u_1, u_2\}$. 
  \item[{\em k} = 4] this case is handled similarly as the case $k=3$.
 \end{description}

At the end of Step 1, $T$ has at most six vertices or all the leaves of $T'$ are anchors. In the former case, we add edges to $M$ as will be described in Section \ref{base-cases} and after which the algorithm terminates. In the latter case we go to Step 2.

\subsubsection{Step 2}

In this step we deal with the case that $T$ has more than six vertices and all the leaves of $T'$ are anchors. We define the {\em second level skeleton tree} ${T''}$ to be the skeleton tree of $T'$. In other words, $T''$ is the tree which is obtained from $T'$ by removing all the leaves. For clarity we use $w$ to refer to a leaf of $T''$, and we use $w$, $w'$, and $w''$, respectively, to refer to the copies of vertex $w$ in $T$, $T'$, and $T''$. For now suppose that $T''$ has at least two vertices. Consider a leaf $w''$ and its neighbor $y''$ in $T''$. Note that in $T$, $w$ is connected to $y$, to at least one anchor, and possibly to some leaves of $T$. After Step 1, the copy of $w''$ in $T'$, i.e. $w'$, is connected to anchors $v'_1,\dots,v'_k$ in $T'$ (or $v_1,\dots,v_k$ in $T$) for some $1 \le k \le 4$, and connected to at most $4-k$ leaves of $T$. In $T$, we consider the clockwise ordering of the non-leaf neighbors of $w$. Let this ordering be $y, v_1, v_2, \dots, v_k$. We denote the pair of leaves connected to anchor $v_i$ by $a_i$ and $b_i$ in clockwise order around $v_i$; see Figure \ref{anchor-orientation}. 

\begin{figure}[ht]
  \centering
    \includegraphics[width=0.6\textwidth]{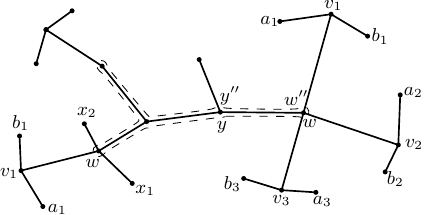}
  \caption{Second level skeleton tree $T''$ is surrounded by dashed line. $v_i$'s are ordered clockwise around leaf vertex $w''$, as well as $x_i$'s. $a_i$ and $b_i$ are ordered clockwise around $v_i$.}
\label{anchor-orientation}
\end{figure}

In this step we pick an arbitrary leaf $w''\in T''$ and according to the number of anchors incident to $w''$, i.e. $k$, we 
add edges to $M$. Since $1 \le k \le 4$, four cases occur and we discuss each case separately. Before that, we state some lemmas.

\begin{figure}[ht]
  \centering
    \includegraphics[width=0.3\textwidth]{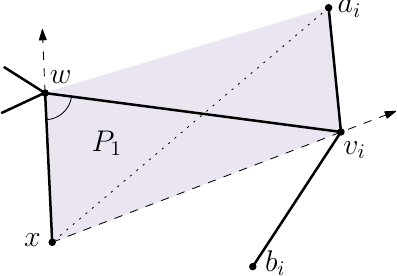}
  \caption{Illustrating Lemma \ref{empty-quadrilateral-lemma}.}
\label{empty-quadrilateral-fig}
\end{figure}

\begin{lemma}
\label{empty-quadrilateral-lemma}
Let $w''$ be a leaf in $T''$. Consider an anchor $v_i$ which is adjacent to $w$ in $T$. For any neighbor $x$ of $w$ for which $x\neq v_i$, if $\cw(v_iwx)\le\pi/2$ (resp. $\cw(xwv_i)\le\pi/2$), the polygon $P_1=\{v_i,x,w,a_i\}$ (resp. $P_2 = \{v_i,b_i,w,x\}$) is convex and empty.
\end{lemma}
\begin{proof}
We prove the case when $\cw(v_iwx)\le\pi/2$; see Figure \ref{empty-quadrilateral-fig}. The proof for the second case is symmetric. To prove the convexity of $P_1$ we show that the diagonals $\overline{v_iw}$ and $\overline{a_ix}$ of $P_1$ intersect each other. To show the intersection we argue that $a_i$ lies to the left of $\overrightarrow{xv_i}$ and to the right of $\overrightarrow{xw}$.

Consider $\overrightarrow{xv_i}$. According to Lemma \ref{empty-triangle-lemma}, triangle $\bigtriangleup v_ixw$ is empty so $b_i$ lies to the right  of $\overrightarrow{xv_i}$. On the other hand, $v_i$ is an anchor, so $\cw(a_iv_ib_i) \ge \pi$, and hence $a_i$ lies to the left of $\overrightarrow{xv_i}$.
Now consider $\overrightarrow{xw}$. For the sake of contradiction, suppose that $a_i$ is to the left of $\overrightarrow{xw}$. Since $\cw(v_iwx)\le\pi/2$, the angle $\cw(a_iwv_i)$ is greater than $\pi/2$. This means that $\overline{v_ia_i}$ is the largest side of $\bigtriangleup a_iwv_i$, which contradicts that $\overline{v_ia_i}$ is an edge of $MST(P)$. So $a_i$ lies to the right of $\overrightarrow{xw}$. Therefore, $\overline{v_iw}$ intersects $\overline{a_ix}$ and $P_1$ is convex. $P_1$ is empty because by Lemma \ref{empty-triangle-lemma}, the triangles $\bigtriangleup v_iwx$ and $\bigtriangleup v_iwa_i$ are empty. 
\end{proof}

\begin{figure}[ht]
  \centering
    \includegraphics[width=0.4\textwidth]{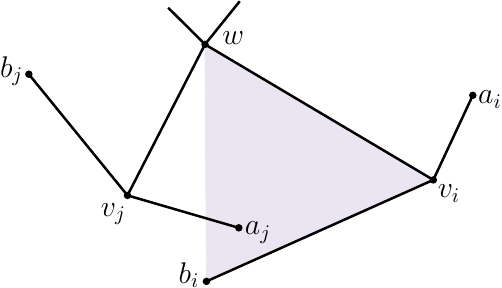}
  \caption{Proof of Lemma \ref{clockwise-order-lemma}.}
\label{clockwise-order-fig}
\end{figure}

\begin{lemma}
\label{clockwise-order-lemma}
Let $w''$ be a leaf in $T''$ and consider the clockwise sequence $v_1,\dots,v_k$ of anchors that are incident on $w$. The sequence of vertices $a_1, v_1, b_1, \dots,a_k,v_k,b_k$ are angularly sorted in clockwise order around $w$. 
\end{lemma}
\begin{proof}
 Using contradiction, consider two vertices $v_i$ and $v_j$, and assume that $v_i$ comes before $v_j$ and $a_j$ comes before $b_i$ in the clockwise order; see Figure \ref{clockwise-order-fig}. Either $a_j$ is in $\bigtriangleup wv_ib_i$ or $b_i$ is in $\bigtriangleup wv_ja_j$. However, by Lemma \ref{empty-triangle-lemma}, neither of these two cases can happen.
\end{proof}

\begin{figure}[ht]
  \centering
\setlength{\tabcolsep}{0in}
  $\begin{tabular}{cc}
  \multicolumn{1}{m{.5\columnwidth}}{\centering\includegraphics[width=.3\columnwidth]{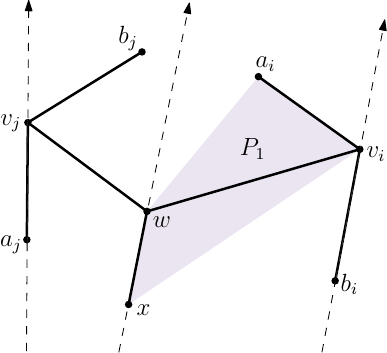}}
  &\multicolumn{1}{m{.5\columnwidth}}{\centering\includegraphics[width=.3\columnwidth]{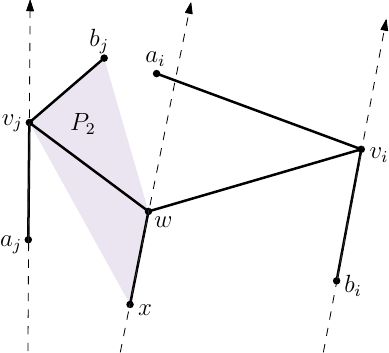}}
  \\
  (a) & (b)
  \end{tabular}$
  \caption{Illustrating Lemma \ref{anchor2-vertex1-lemma}.}
\label{anchor2-vertex1-1-fig}
\end{figure}

\begin{lemma}
\label{anchor2-vertex1-lemma}
Let $w''$ be a leaf in $T''$ and consider the clockwise sequence $v_1,\dots,v_k$ of anchors that are adjacent to $w$. Let $1\le i < j \le k$ and let $x$ be a neighbor of $w$ for which $x\neq v_i$ and $x\neq v_j$. 
\begin{enumerate}
 \item If $x$ is between $v_i$ and $v_j$ in the clockwise order:
    \begin{enumerate}
      \item if $a_i$ is to the right of $\overrightarrow{xw}$, then $P_1=\{x, w, a_i, v_i\}$ is convex and empty.
      \item if $a_i$ is not to the right of $\overrightarrow{xw}$, then $P_2=\{w, x,v_j, b_j\}$ is convex and empty.
    \end{enumerate}
  \item If $v_j$ is between $v_i$ and $x$, or $v_i$ is between $x$ and $v_j$ in the clockwise order:
    \begin{enumerate}
      \item if $b_i$ is to the left of $\overrightarrow{xw}$, then $P_3=\{w, x, v_i, b_i\}$ is convex and empty.
      \item if $b_i$ is not to the left of $\overrightarrow{xw}$, then $P_4=\{x, w, a_j, v_j\}$ is convex and empty.
    \end{enumerate}
\end{enumerate}
\end{lemma}

\begin{proof}
We only prove the first case, the proof for the second case is symmetric. Thus, we assume that $x$ is between $v_i$ and $v_j$ in the clockwise order. First assume that $a_i$ is to the right of $\overrightarrow{xw}$. See Figure \ref{anchor2-vertex1-1-fig}(a). Consider $\overrightarrow{b_iv_i}$. Since $v_i$ is an anchor, $a_i$ cannot be to the right of $\overrightarrow{b_iv_i}$, and according to Lemma \ref{empty-triangle-lemma}, $x$ cannot be to the right of $\overrightarrow{b_iv_i}$. For the same reasons, both the vertices $b_j$ and $x$ cannot be to the left of $\overrightarrow{a_jv_j}$. Now consider $\overrightarrow{xw}$. By assumption, $a_i$ is to the right of $\overrightarrow{xw}$. Therefore $\overline{xa_i}$ intersects $\overline{wv_i}$ and hence $P_1$ is convex.

Now assume that $a_i$ is not to the right of $\overrightarrow{xw}$; see Figure \ref{anchor2-vertex1-1-fig}(b). By Lemma \ref{clockwise-order-lemma}, $b_j$ is to the left of $\overrightarrow{xw}$. Therefore, $\overline{xb_j}$ intersects $\overline{wv_j}$ and hence $P_2$ is convex. 
The emptiness of the polygons follows directly from Lemma \ref{empty-triangle-lemma}.
\end{proof}

\begin{figure}[ht]
  \centering
    \includegraphics[width=0.35\textwidth]{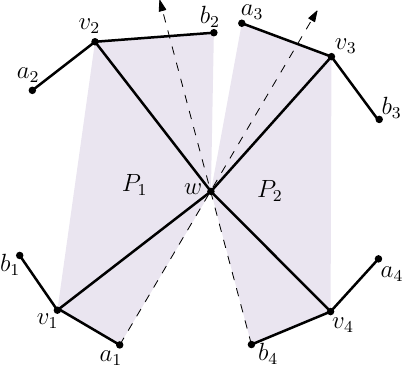}
  \caption{Illustrating Lemma \ref{separation-lemma}.}
\label{anchor4-5-fig}
\end{figure}

\begin{lemma}
\label{separation-lemma}
Let $w''$ be a leaf in $T''$ and consider the clockwise sequence $v_1,\dots,v_k$ of anchors that are incident on $w$. If $4\le k \le 5$, then at least one of the pentagons $P_1=\{a_1,v_1,v_2,b_2,w\}$ and $P_2=\{a_{k-1},v_{k-1},v_k,b_k,w\}$ is convex and empty.
\end{lemma}
\begin{proof}
We prove the case when $k=4$; the proof for $k=5$ is analogous. By Lemma \ref{clockwise-order-lemma}, $a_1$ comes before $b_4$ in the clockwise order. Consider $\overrightarrow{a_1w}$; see Figure \ref{anchor4-5-fig}. Let $l(a_1, b_2)$ denote the line through $a_1$ and $b_2$. If $b_2$ is to the left of $\overrightarrow{a_1w}$, then $l(a_1, b_2)$ separates $w$ from $v_1$ and $v_2$. Otherwise, $b_2$ is to the right of $\overrightarrow{a_1w}$. Now consider $\overrightarrow{b_4w}$. If $a_3$ is to the right of $\overrightarrow{b_4w}$, then $l(a_3, b_4)$ separates $w$ from $v_3$ and $v_4$. The remaining case, i.e., when $a_3$ is to the left of $\overrightarrow{b_4w}$, cannot happen by Lemma \ref{clockwise-order-lemma} (because, otherwise, $a_3$ would come before $b_2$ in the clockwise order).

Thus, we have shown that (i) $l(a_1, b_2)$ separates $w$ from $v_1$ and $v_2$ or (ii) $l(a_3, b_4)$ separates $w$ from $v_3$ and $v_4$. Assume w.l.o.g. that (i) holds. Now to prove the convexity of $P_1$ we show that all internal angles of $P_1$ are less than $\pi$. Since $v_1$ is an anchor, $\cw(b_1v_1a_1)\le \pi$. By Lemma~\ref{empty-triangle-lemma}, $b_1$ is to the left of of $\overrightarrow{v_1v_2}$. Therefore, $\cw(v_2v_1a_1) < \cw(b_1v_1a_1) \le \pi$. On the other hand $cw(wv_1a_1)\ge \pi/3$, so in $\bigtriangleup a_1v_1w$, $\cw(v_1a_1w) \le 2\pi/3$. By a similar analysis $\cw(b_2v_2v_1)$ and $\cw(wb_2v_2)$ are less than $\pi$. In addition, $\cw(a_1wb_2)$ in $\bigtriangleup a_1wb_2$ is less than $\pi$. Thus, $P_1$ is convex. Its emptiness is assured from the emptiness of the triangles $\bigtriangleup v_2wb_2$, $\bigtriangleup v_1wa_1$ and $\bigtriangleup v_1wv_2$.
\end{proof}

Now we are ready to present the details of Step 2. Recall that $T$ has more than six vertices and all leaves of $T'$ are anchors. In this case $T''$ has at least one vertex. If $T''$ has exactly one vertex, we add edges to $M$ as will be described in Section \ref{base-cases} and after which the algorithm terminates. Assume that $T''$ has at least two vertices. Pick a leaf $w''$ in $T''$. As before, let $v_1,\dots,v_k$ where $1\le k \le 4$ be the clockwise order of the anchors connected to $w$. Let $x_1,\dots,x_{\ell}$ where $\ell \le 4-k\le3$ be the clockwise order of the leaves of $T$ connected to $w$; see Figure \ref{anchor-orientation}. This means that $deg(w)=1+k+\ell$, where $k+\ell \le 4$. Now we describe different configurations that may appear at $w$, according to $k$ and $\ell$.
 \begin{description}
  \item[Case 1:] assume $k=1$. If $\ell=0$ or $1$, add $(a_1,v_1)$ and $(b_1,w)$ to $M$ and set $T=T\setminus\{a_1,\allowbreak b_1,\allowbreak v_1,\allowbreak w\}$. If $\ell=1$ remove $x_1$ from $T$ as well. If $\ell=2$, consider $\alpha_i$ as the acute angle between segments $\overline{wv_1}$ and $\overline{wx_i}$. W.l.o.g. assume $\alpha_1 = \min(\alpha_1,\alpha_2)$. Two cases may arise: (i) $\alpha_1\le \pi/2$, (ii) $\alpha_1> \pi/2$. If (i) holds, w.l.o.g. assume that $x_1$ comes before $v_1$ in clockwise order around $w$. According to Lemma \ref{empty-quadrilateral-lemma} polygon $P=\{v_1, b_1, w, x_1\}$ is convex and empty. So we add $(v_1,a_1)$, $(b_1,x_1)$ and $(x_2,w)$ to $M$. Other cases are handled in a similar way. If (ii) holds, according to Lemma \ref{empty-triangle-lemma} triangle $\bigtriangleup x_1wx_2$ is empty. So we add $(v_1, a_1)$, $(b_1,w)$ and $(x_1,x_2)$ to $M$. In both cases set $T=T\setminus\{a_1,b_1,v_1, x_1, x_2, w\}$. If $\ell=3$, remove $x_3$ from $T$ and handle the rest as $\ell=2$.
  \item[Case 2:] assume $k=2$. If $\ell = 0$, we add $(v_1, a_1)$, $(b_1, w)$, and $(v_2, a_2)$ to $M$. If $\ell = 1$ suppose that $x_1$ comes before $v_1$ in clockwise ordering. According to Lemma \ref{anchor2-vertex1-lemma} one of polygons $P_3$ and $P_4$ is empty; suppose it be $P_3=\{w,x_1,v_1,b_1\}$ (where $i=1$ and $j=2$ in Lemma \ref{anchor2-vertex1-lemma}). Thus we add $(v_1, a_1)$, $(b_1,x_1)$, $(v_2,a_2)$ and $(b_2,w)$ to $M$. In both cases set $T=T\setminus\{a_1,\allowbreak b_1,\allowbreak v_1, \allowbreak a_2,b_2,v_2,\allowbreak w\}$ and if $\ell = 1$ remove $x_1$ from $T$ as well. If $\ell = 2$, remove $x_2$ from $T$ and handle the rest as $\ell=1$.
  \item[Case 3:] assume $k=3$. If $\ell = 0$ then set $M=M\cup\{(v_1,a_1), (b_1,w), (v_2,a_2), (v_3,a_3)\}$. If $\ell=1$, consider $\beta_i$ as the acute angle between segments $\overline{wx_1}$ and $\overline{wv_i}$. W.l.o.g. assume $\beta_2$ has minimum value among all $\beta_i$'s and $x_1$ comes after $v_2$ in clockwise order. According to Lemma \ref{anchor2-vertex1-lemma} one of polygons $P_1$ and $P_2$ is empty; suppose it be $P_1=\{w,x_1,v_2,a_2\}$ (where $i=2$ and $j=3$ in Lemma \ref{anchor2-vertex1-lemma}). Thus we set $M=M\cup\{(v_2, b_2), (a_2,x_1),\allowbreak (v_1,a_1), (b_1,w),\allowbreak  (v_3,a_3)\}$. In both cases set $T = T\setminus\{a_1,b_1,\allowbreak v_1, a_2,\allowbreak b_2,v_2, \allowbreak a_3,b_3,v_3, w\}$ and if $\ell = 1$ remove $x_1$ from $T$ as well. Other cases can be handled similarly.

\begin{figure}[ht]
  \centering
    \includegraphics[width=0.35\textwidth]{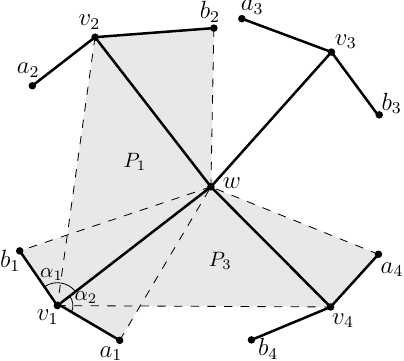}
  \caption{The vertex $v_1$ is shared between $P_1$ and $P_3$.}
\label{case4-fig}
\end{figure}

  \item[Case 4:] assume $k=4$. According to Lemma \ref{separation-lemma} one of $P_1=\{a_1,v_1,v_2,b_2,w\}$ and $P_2=\{a_3,v_3,v_4,b_4,w\}$ is convex and empty. Again by Lemma \ref{separation-lemma} one of $P_3=\{a_4,v_4,v_1,b_1,w\}$ and $P_4=\{a_2,v_2,v_3,b_3,w\}$ is also convex and empty. Without loss of generality assume that $P_1$ and $P_3$ are empty. See Figure \ref{case4-fig}. Clearly, these two polygons share a vertex ($v_1$ in Figure \ref{case4-fig}). Let $\alpha_1=\cw(b_1,v_1,w)$ which is contained in $P_3$ and $\alpha_2=\cw(w,v_1,a_1)$ which is contained in $P_1$. We pick one of the polygons $P_1$ and $P_3$ which minimizes $\alpha_i, i=1,2$. Let $P_1$ be that polygon. So we set $M = M \cup  \allowbreak \{(v_1,b_1) , \allowbreak (v_2,a_2) ,\allowbreak  (a_1,b_2),\allowbreak (v_3, a_3),\allowbreak  (b_3,w),\allowbreak  (v_4,a_4)\}$ and set $T = T\setminus\{a_1,b_1,\allowbreak v_1,\allowbreak  a_2,\allowbreak b_2,v_2, \allowbreak  a_3, b_3,v_3,\allowbreak a_4,\allowbreak b_4\allowbreak,v_4,w \}$. 
 \end{description}

This concludes Step 2. Go back to Step 1.

\subsubsection{Base Cases}
\label{base-cases}

In this section we describe the base cases of our algorithm. As mentioned in Steps 1 and 2, we may have two base cases: (a) $T$ has at most $t\le 6$ vertices, (b) $T''$ has only one vertex.

\begin{paragraph}{(a) {\em t}${\bf \le 6}$} Suppose that $T$ has at most six vertices.
\begin{description}
  \item [{\em t} = 2] it can happen only if $t=n=2$, and we add the only edge to $M$.
  \item [{\em t} = 4, 5] in this case we match four vertices. If $t=4$, $T$ could be a star or a path of length three, and in both cases we match all the vertices. If $t=5$, remove one of the leaves and match other four vertices.
  \item [{\em t} = 6] in this case we match all the vertices. If one of the leaves connected to a vertex of degree two, we match those two vertices and handle the rest as the case when $t=4$, otherwise, each leaf of $T$ is connected to a vertex of degree more than two, and hence $T'$ has at most two vertices. Figure \ref{n-small}(a) shows the solution for the case when $T'$ has only one vertex and $T$ is a star; note that at least two angles are less than $\pi$. 
  Now consider the case when $T'$ has two vertices, $v_1$ and $v_2$, which have degree three in $T$. Figure \ref{n-small}(b) shows the solution for the case when neither $v_1$ nor $v_2$ is an anchor. 
  Figure \ref{n-small}(c) shows the solution for the case when $v_1$ is an anchor but $v_2$ is not.
  Figure \ref{n-small}(d) shows the solution for the case when both $v_1$ and $v_2$ are anchors. Since $v_2$ is an anchor in Figure \ref{n-small}(d), at least one of $\cw(b_2v_2v_1)$ and $\cw(v_1v_2a_2)$ is less than or equal to $\pi/2$. W.l.o.g. assume $\cw(v_1v_2a_2)\le\pi/2$. By Lemma \ref{empty-quadrilateral-lemma} polygon $P=\{v_1,a_2,v_2,a_1\}$ is convex and empty. We add $(v_1,b_1)$, $(v_2,b_2)$, and $(a_1,a_2)$ to $M$.

  \item [{\em t} = 1, 3] this case could not happen. Initially $t=n$ is even. Consider Step 1; before each iteration $t$ is bigger than six and during the iteration two vertices are removed from $T$. So, at the end of Step 1, $t$ is at least five. Now consider Step 2; before each iteration $T''$ has at least two vertices and during the iteration at most one vertex is removed from $T''$. So, at the end of Step 2, $T''$ has at least one vertex that is connected to at least one anchor. This means that $t$ is at least four. Thus, $t$ could never be one or three before and during the execution of the algorithm.

\begin{figure}[ht]
  \centering
\setlength{\tabcolsep}{0in}
  $\begin{tabular}{cccc}
  \multicolumn{1}{m{.25\columnwidth}}{\centering\includegraphics[width=.20\columnwidth]{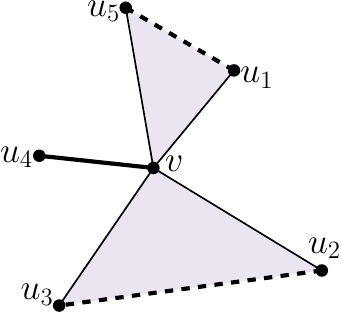}}
  &\multicolumn{1}{m{.25\columnwidth}}{\centering\includegraphics[width=.24\columnwidth]{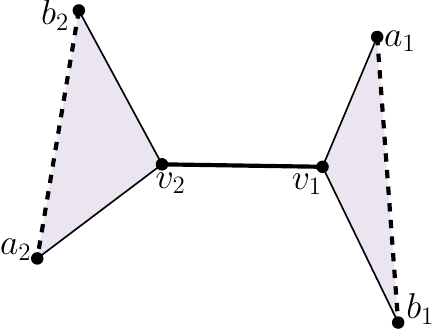}}
  &\multicolumn{1}{m{.25\columnwidth}}{\centering\includegraphics[width=.23\columnwidth]{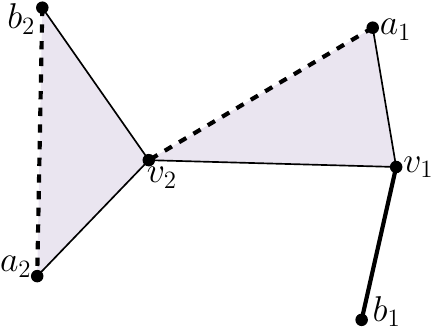}}
  &\multicolumn{1}{m{.25\columnwidth}}{\centering\includegraphics[width=.23\columnwidth]{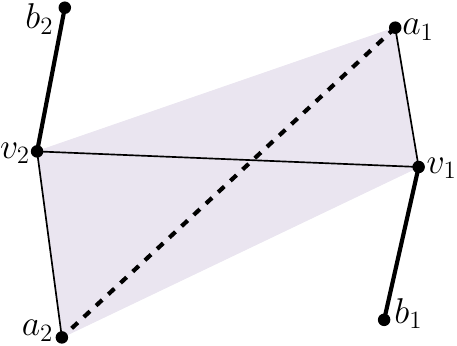}}\\
  (a) & (b)&(c)&(d)
  \end{tabular}$
  \caption{The bold (solid and dashed) edges are added to $M$ and all vertices are matched. (a) a star, (b) no anchor, (c) one anchor, and (d) two anchors.}
\label{n-small}
\end{figure} 

\end{description}
\end{paragraph}

\begin{paragraph}{(b) $T''$ has one vertex}
In this case, the only vertex $w''\in T''$ is connected to at least two anchors, otherwise $w''$ would have been matched in Step 1. So we consider different cases when $w$ is connected to $k$, $2\le k \le 5$ anchors and $\ell \le 5-k$ leaves of $T$:
\begin{description}
 \item[{\em k} = 2] if $\ell=0,1,2$ we handle it as Case 2 in Step 2. If $\ell=3$, at least two leaves are consecutive, say $x_1$ and $x_2$. Since $cw(x_1wx_2) <\pi$ we add $(x_1,x_2)$ to $M$ and handle the rest like the case when $\ell=1$. 
\item[{\em k} = 3] if $\ell=2$ remove $x_2$ from $T$. Handle the rest as Case 3 in Step 2.
\item[{\em k} = 4] if $\ell=1$ remove $x_1$ from $T$. Handle the rest as Case 4 in Step 2. 
\item[{\em k} = 5] add $(v_5, a_5)$ to $M$, remove $a_5$, $b_5$, $v_5$ from $T$, and handle the rest as Case 4 in Step 2.
\end{description}
\end{paragraph}

This concludes the algorithm.

\begin{lemma}
\label{convex-disjoint-lemma}
The convex empty regions that are considered in different iterations of the algorithm, do not overlap.
\end{lemma}
\begin{proof}
In Step 1, Step 2 and the base cases, we used three types of convex empty regions; see Figure \ref{convex-disjoint-fig}.
Using contradiction, suppose that two convex regions $P_1$ and $P_2$ overlap. Since the regions are empty, no vertex of $P_1$ is in the interior of $P_2$ and vice versa. Then, one of the edges in $MST(P)$ that is shared with $P_1$ intersects some edge in $MST(P)$ that is shared with $P_2$, which is a contradiction.   
\end{proof}

\begin{figure}[ht]
  \centering
\setlength{\tabcolsep}{0in}
  $\begin{tabular}{ccc}
  \multicolumn{1}{m{.33\columnwidth}}{\centering\includegraphics[width=.20\columnwidth]{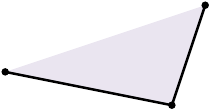}}
  &\multicolumn{1}{m{.33\columnwidth}}{\centering\includegraphics[width=.26\columnwidth]{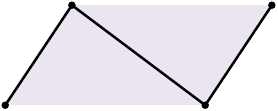}}
  &\multicolumn{1}{m{.33\columnwidth}}{\centering\includegraphics[width=.32\columnwidth]{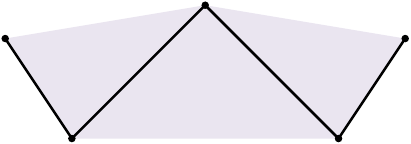}}
  \end{tabular}$
  \caption{Empty convex regions. Bold edges belong to $MST(P)$.}
\label{convex-disjoint-fig}
\end{figure}

\begin{theorem}
Let $P$ be a set of $n$ points in the plane, where $n$ is even, and let $\btopt$ be the minimum bottleneck length of any plane perfect matching of $P$. In $O(n\log n)$ time, a plane matching of $P$ of size at least $\frac{2n}{5}$ can be computed, whose bottleneck length is at most $(\sqrt{2}+\sqrt{3})\btopt$.
\end{theorem}
\begin{proof}

 {\em Proof of planarity}: In each iteration, in Step 1, Step 2, and in the base cases, the edges added to $M$ are edges of $MST(P)$ or edges inside convex empty regions. By Lemma \ref{convex-disjoint-lemma} the convex empty regions in each iteration will not be intersected by the convex empty regions in the next iterations. Therefore, the edges of $M$ do not intersect each other and $M$ is plane.

{\em Proof of matching size}: In Step 1, in each iteration, all the vertices which are excluded from $T$ are matched. In Step 2, when $k=\ell=1$ we match four vertices out of five, and when $k=3, \ell =0$ we match eight vertices out of ten. In base case (a) when $t=5$ we match four vertices out of five. In base case (b) when $k=3, \ell =0$ we match eight vertices out of ten. In all other cases of Step 2 and the base cases, the ratio of matched vertices is more than $4/5$. Thus, in each iteration at least $4/5$ of the vertices removed from $T$ are matched and hence $|M_i|\ge\frac{2n_i}{5}$. Therefore, 
$$
 |M|= \sum_{i=1}^{k} |M_i| \ge \sum_{i=1}^{k} \frac{2n_i}{5} = \frac{2n}{5}.
$$

{\em Proof of edge length}: By Lemma \ref{longest-edge} the length of edges of $T$ is at most $\btopt$. Consider an edge $e\in M$ and the path $\delta$ between its end points in $T$. If $e$ is added in Step 1, then $|e|\le 2\btopt$ because $\delta$ has at most two edges. If $e$ is added in Step 2, $\delta$ has at most three edges ($|e|\le 3\btopt$) except in Case 4. In this case we look at $\delta$ in more detail. We consider the worst case when all the edges of $\delta$ have maximum possible length $\btopt$ and the angles between the edges are as big as possible; see Figure \ref{path-fig}. Consider the edge $e=(a_1,b_2)$ added to $M$ in Case 4. Since $v_2$ is an anchor and $\cw(wv_2a_2)\ge \pi/3$, the angle $\cw(b_2v_2w)\le 2\pi/3$. As our choice between $P_1$ and $P_3$ in Case 4, $\cw(w,v_1,a_1)\le\pi/2$. Recall that $e$ avoids $w$, and hence $|e|\le(\sqrt{2}+\sqrt{3})\btopt$. The analysis for the base cases is similar.

\begin{figure}[ht]
  \centering
    \includegraphics[width=0.55\textwidth]{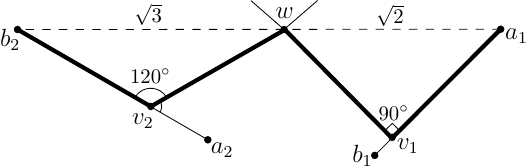}
  \caption{Path $\delta$ (in bold) with four edges of length $\btopt$ between end points of edge $(a_1,b_2)$.}
\label{path-fig}
\end{figure}

{\em Proof of complexity}: The Delaunay triangulation of $P$ can be computed in $O(n \log n)$ time. Using Kruskal's algorithm, the forest $F$ of even components can be computed in $O(n\log n)$ time. 
In Step 1 (resp. Step 2) in each iteration, we pick a leaf of $T'$ (resp. $T''$) and according to the number of leaves (resp. anchors) connected to it we add some edges to $M$. Note that in each iteration we can update $T$, $T'$ and $T''$ by only checking the two hop neighborhood of selected leaves. Since the two hop neighborhood is of constant size, we can update the trees in $O(1)$ time in each iteration. 
Thus, the total running time of Step 1, Step 2, and the base cases is $O(n)$ and the total running time of the algorithm is $O(n\log n)$.
\end{proof}

\section{Conclusion}
\label{conclusion}
We considered the NP-hard problem of computing a bottleneck plane perfect matching of a point set. Abu-Affash et al. \cite{Abu-Affash2014} presented a $2\sqrt{10}$-approximation for this problem. We used the maximum plane matching problem in unit disk graphs (UDG) as a tool for approximating a bottleneck plane perfect matching. In Section \ref{three-approximation} we presented an algorithm which computes a plane matching of size $\frac{n}{6}$ in UDG, but it is still open to show if this algorithm terminates in polynomial number of steps or not. We also presented a $\frac{2}{5}$-approximation algorithm for computing a maximum matching in UDG. By extending this algorithm we showed how one can compute a bottleneck plane matching of size $\frac{n}{5}$ with edges of optimum-length. A modification of this algorithm gives us a plane matching of size at least $\frac{2n}{5}$ with  edges of length at most $\sqrt{2}+\sqrt{3}$ times the optimum. 

\bibliographystyle{abbrv}
\bibliography{BNCM}
\end{document}